\documentclass[letterpaper, 10pt, conference]{ieeeconf}
\let\proof\relax   

\usepackage{amsthm,xpatch}
\usepackage{amsmath,amsfonts}
\usepackage{cite}
\usepackage{amssymb}
\usepackage{dsfont}
\usepackage{graphicx, subfigure}
\usepackage{color}
\usepackage{breqn}
\usepackage{mathtools}
\usepackage{bbm}
\usepackage{latexsym}
\usepackage[ruled, linesnumbered]{algorithm2e}
\usepackage{accents}
\usepackage{tikz}
\usepackage[mathscr]{eucal}
\usepackage{caption}
\usepackage{booktabs}
\usetikzlibrary{calc}
\usetikzlibrary{math}
\usetikzlibrary{arrows.meta}
\usepackage{csquotes}
\usepackage{varioref}
\newtheorem{lemma}{Lemma}
\newtheorem{theorem}{Theorem}

\newtheorem{definition}{Definition}
\newtheorem{example}{Example}

\newcommand{\abs}[1]{\lvert{#1}\rvert}
\newcommand{\mbf}[1]{\mathbf{#1}}
\newcommand*{\transpose}{%
  {\mathpalette\@transpose{}}%
}

\IEEEoverridecommandlockouts

\begin{document}

\newcommand{\SB}[3]{
\sum_{#2 \in #1}\biggl|\overline{X}_{#2}\biggr| #3
\biggl|\bigcap_{#2 \notin #1}\overline{X}_{#2}\biggr|
}

\newcommand{\Mod}[1]{\ (\textup{mod}\ #1)}

\newcommand{\overbar}[1]{\mkern 0mu\overline{\mkern-0mu#1\mkern-8.5mu}\mkern 6mu}

\makeatletter
\newcommand*\nss[3]{%
  \begingroup
  \setbox0\hbox{$\m@th\scriptstyle\cramped{#2}$}%
  \setbox2\hbox{$\m@th\scriptstyle#3$}%
  \dimen@=\fontdimen8\textfont3
  \multiply\dimen@ by 4             
  \advance \dimen@ by \ht0
  \advance \dimen@ by -\fontdimen17\textfont2
  \@tempdima=\fontdimen5\textfont2  
  \multiply\@tempdima by 4
  \divide  \@tempdima by 5          
  \ifdim\dimen@<\@tempdima
    \ht0=0pt                        
    \@tempdima=\fontdimen5\textfont2
    \divide\@tempdima by 4          
    \advance \dimen@ by -\@tempdima 
    \ifdim\dimen@>0pt
      \@tempdima=\dp2
      \advance\@tempdima by \dimen@
      \dp2=\@tempdima
    \fi
  \fi
  #1_{\box0}^{\box2}%
  \endgroup
  }
\makeatother

\makeatletter
\renewenvironment{proof}[1][\proofname]{\par
  \pushQED{\qed}%
  \normalfont \topsep6\p@\@plus6\p@\relax
  \trivlist
  \item[\hskip\labelsep
        \itshape
    #1\@addpunct{:}]\ignorespaces
}{%
  \popQED\endtrivlist\@endpefalse
}
\makeatother

\makeatletter
\newsavebox\myboxA
\newsavebox\myboxB
\newlength\mylenA

\newcommand*\xoverline[2][0.75]{%
    \sbox{\myboxA}{$\m@th#2$}%
    \setbox\myboxB\null
    \ht\myboxB=\ht\myboxA%
    \dp\myboxB=\dp\myboxA%
    \wd\myboxB=#1\wd\myboxA
    \sbox\myboxB{$\m@th\overline{\copy\myboxB}$}
    \setlength\mylenA{\the\wd\myboxA}
    \addtolength\mylenA{-\the\wd\myboxB}%
    \ifdim\wd\myboxB<\wd\myboxA%
       \rlap{\hskip 0.5\mylenA\usebox\myboxB}{\usebox\myboxA}%
    \else
        \hskip -0.5\mylenA\rlap{\usebox\myboxA}{\hskip 0.5\mylenA\usebox\myboxB}%
    \fi}
\makeatother

\xpatchcmd{\proof}{\hskip\labelsep}{\hskip3.75\labelsep}{}{}

\pagestyle{plain}

\title{\fontsize{22.59}{28}\selectfont Sparse Graph Codes for Non-adaptive \\ Quantitative Group Testing}

\author{Esmaeil Karimi, Fatemeh Kazemi, Anoosheh Heidarzadeh, Krishna R. Narayanan, and Alex Sprintson\thanks{The authors are with the Department of Electrical and Computer Engineering, Texas A\&M University, College Station, TX 77843 USA (E-mail: \{esmaeil.karimi, fatemeh.kazemi, anoosheh, krn, spalex\}@tamu.edu).}}


\maketitle

\begin{abstract}
This paper considers the problem of Quantitative Group Testing (QGT). Consider a set of $N$ items among which $K$ items are defective. The QGT problem is to identify (all or a sufficiently large fraction of) the defective items, where the result of a test reveals the number of defective items in the tested group.
In this work, we propose a non-adaptive QGT algorithm using sparse graph codes over bi-regular bipartite graphs with left-degree $\ell$ and right degree $r$ and binary $t$-error-correcting BCH codes. The proposed scheme provides exact recovery with probabilistic guarantee, i.e. recovers all the defective items with high probability. In particular, we show that for the sub-linear regime where $\frac{K}{N}$ vanishes as $K,N\rightarrow\infty$, the proposed algorithm requires at most ${m=c(t)K\left(t\log_2\left(\frac{\ell N}{c(t)K}+1\right)+1\right)+1}$ tests to recover all the defective items with probability approaching one as ${K,N\rightarrow\infty}$, where $c(t)$ depends only on $t$. The results of our theoretical analysis reveal that the minimum number of required tests is achieved by $t=2$. The encoding and decoding of the proposed algorithm for any $t\leq 4$ have the computational complexity of $\mathcal{O}(K\log^2 \frac{N}{K})$ and $\mathcal{O}(K\log \frac{N}{K})$, respectively. Our simulation results also show that the proposed algorithm significantly outperforms a non-adaptive semi-quantitative group testing algorithm recently proposed by Abdalla \emph{et al.} in terms of the required number of tests for identifying all the defective items with high probability.
\end{abstract}

\section{introduction}
In this work, we consider the problem of Quantitative Group Testing (QGT). Consider a set of $N$ items among which $K$ items are defective. The QGT problem is to identify (all or a sufficiently large fraction of) the defective items, where the result of a test reveals the number of defective items in the tested group. The key difference between the QGT problem and the original group testing problem is that, unlike the former, in the latter the result of each test is either $1$ or $0$ depending on whether the tested group contains any defective items or not. The objective of QGT is to design a test plan with minimum number of tests that identifies (all or a sufficiently large fraction of) the defective items. 

There are two general categories of test strategies: \emph{non-adaptive} and \emph{adaptive}. In an adaptive scheme, each test depends on the outcomes of the previous tests. On the other hand, in a non-adaptive scheme, all tests are planned in advance. In other words, the result of one test does not affect the design of another test. Although, in general, adaptive algorithms require fewer tests, in most practical applications non-adaptive algorithms are preferred since they allow one to perform all tests at once in parallel.


Let $S$ be the index set of the defective items and $\hat{S}$ be an estimation of $S$. Depending on the application at hand, there can be different requirements for the \emph{closeness} of $\hat{S}$ to $S$ \cite{7953326,DBLP:journals/corr/LeePR15}. The strongest condition for closeness is \emph{exact recovery} when it is required that $\hat{S}=S$. Two weaker conditions are \emph{partial recovery without false detections} when it is required that $\hat{S} \subseteq S$ and $|\hat{S}|\geq (1-\epsilon)|S|$, and \emph{partial recovery without missed detections} when it is required that $S \subseteq \hat{S}$ and $|\hat{S}|\leq (1+\epsilon)|S|$. There are also different types of the \emph{recovery guarantees}~\cite{DBLP:journals/corr/LeePR15}. The strongest guarantee is \emph{perfect recovery guarantee} when the exact or partial recovery needs to be achieved with probability $1$ (over the space of all problem instances). A slightly weaker guarantee is \emph{probabilistic recovery guarantee} when it suffices to achieve the exact or partial recovery with high probability only (and not necessarily with probability $1$). In this work, we are interested in the exact recovery of all defective items with the probabilistic recovery guarantee. 


\subsection{Related Work and Applications}
The QGT problem has been extensively studied for a wide range of applications, e.g., multi-access communication, spectrum sensing, and network tomography, see, e.g.,~\cite{WZC:17,heidarzadehfast,heidarzadehuser}, and references therein. This problem was first introduced by Shapiro in~\cite{S:60}. 
Several non-adaptive and adaptive QGT strategies have been previously proposed, see, e.g.,~\cite{B:09,WZC:17,8437774}.  
It was shown in~\cite{L:75} that any non-adaptive algorithm must perform at least $(2K\log_2 (N/K))/\log_2 K$ tests. 
Various order optimal or near-optimal non-adaptive strategies were previously proposed, see, e.g.,~\cite{L:75,8437774,B:09}. The best known polynomial-time non-adaptive algorithms require $K\log N$ tests~\cite{lindstrom1972b2,L:75}. 
Recently, a semi-quantitative group testing scheme based on sparse graph codes was proposed in~\cite{8335478}, where the result of each test is an integer in the set $\{0,1,2,\dots,L\}$. This strategy identifies a $(1-\epsilon)$ fraction of defective items using $c(\epsilon,L)K\log_2 N$ tests with high probability, where $c(\epsilon,L)$ depends only on $\epsilon$ and $L$.  

\subsection{Connection with Compressed Sensing}
A closely related problem to QGT is the problem of compressed sensing (CS) in which the goal is to recover a sparse signal from a set of (linear) measurements. Given an $N$-dimensional sparse signal with a support size up to $K$, the objective is to identify the indices and the values of non-zero elements of the signal with minimum number of measurements. The main differences between the CS problem and the QGT problem are in the signal model and the constraints on the measurement matrix. Most of the existing works on the CS problem consider real-valued signals and measurement matrices. The QGT problem, however, deals with binary signals and requires the measurement matrix to be binary-valued. 

There are a number of CS algorithms in the literature that use binary-valued measurement matrices, see, e.g.~\cite{8368314,IWEN20141} and references therein. However, these strategies either use techniques which are not applicable to binary signals, or provide different types of closeness and guarantee than those required in this work. There are also several CS algorithms for the support recovery where the objective is to determine the indices of the non-zero elements of the signal but not their values~\cite{DBLP:journals/corr/LiPR14,pmlr-v51-scarlett16,5766202}. The support recovery problem is indeed equivalent to the QGT problem. Notwithstanding, the existing schemes for support recovery rely on non-binary measurement matrices, and hence are not suitable for the QGT problem. 

Last but not least, to the best of our knowledge, the majority of works on the CS problem focus mainly on the order optimality of the number of measurements, whereas in this work for the QGT problem we are also interested in minimizing the constant factor hidden in the order. 

\subsection{Main Contributions}
In this work, we propose a non-adaptive quantitative group testing strategy for the sub-linear regime where $\frac{K}{N}$ vanishes as $K,N\rightarrow\infty$. We utilize sparse graph codes over bi-regular bipartite graphs with left-degree $\ell$ and right-degree $r$ and binary $t$-error-correcting BCH codes for the design of the proposed strategy. Leveraging powerful density evolution techniques for the analysis enables us not only to determine the exact value of constants in the number of tests needed but also to provide provable performance guarantees. We show that the proposed scheme provides exact recovery with probabilistic guarantee, i.e. recovers all the defective items with high probability. In particular, for the sub-linear regime, the proposed algorithm requires at most ${m=c(t)K\left(t\log_2\left(\frac{\ell N}{c(t)K}+1\right)+1\right)+1}$ tests to recover all defective items with probability approaching one as ${K,N\rightarrow\infty}$, where $c(t)$ depends only on $t$. 

The results of our theoretical analysis reveal that the minimum number of required tests for the proposed algorithm is achieved by $t=2$. Moreover, for any $t\leq 4$, the encoding and decoding of the proposed algorithm have the computational complexity of $\mathcal{O}(K\log^2 \frac{N}{K})$ and $\mathcal{O}(K\log \frac{N}{K})$, respectively. 


\section{Problem Setup and Notation}\label{sec:SN}
Throughout the paper, we use bold-face small and capital letters to denote vectors and matrices, respectively. 

In this work, we consider the problem of quantitative group testing (QGT) with exact recovery and probabilistic guarantee, defined as follows. Consider a set of $N$ items among which $K$ items are defective. We focus on the sub-linear regime where the ratio $\frac{K}{N}$ vanishes as $K,N\rightarrow\infty$. The problem is to identify all the defective items with high probability while using minimum number of tests on subsets (groups) of the items, where the result of each test shows the number of defective items in the tested group.

Let the vector $\mathbf{x}\in \{0,1\}^N$ represent the set of $N$ items in which the coordinates with value $1$ correspond to the defective items. A non-adaptive group testing problem consisting of $m$ tests can be represented by a measurement matrix ${\textbf{A}\in \{0,1\}^{m\times N}}$, where the $i$-th row of the matrix corresponds to the $i$-th test. That is, the coordinates with value $1$ in the $i$-th row correspond to the items in the $i$-th test.
The results of the $m$ tests are expressed in the test vector $\mathbf{y}\in\{ 0,1,\dots\}^m$, i.e., 
\begin{equation}\label{eq:gtresult}
\mathbf{y}=[y_{1},\cdots,y_{m}]^{\mathsf{T}}=\mathbf{A}\mathbf{x}.
\end{equation}
The goal is to design a testing matrix $\mathbf{A}$ that has a small number of rows (tests), $m$, and can identify with high probability all the defective items given the test vector $\mathbf{y}$.

\section{Proposed Algorithm \label{sec:main results}}
\subsection{Binary $t$-error-correcting codes and $t$-separable matrices}

\begin{definition}\label{def:d-separable}($t$-separable matrix)
A binary matrix ${\mathbf{D}\in\{0,1\}^{m\times n}}$ (for $n>t$) is $t$-separable over field $\mathbb{F}$ if the sum (over field $\mathbb{F}$) of any set of $t$ columns is distinct.
\end{definition}
\begin{example}
Consider the following matrix,
\begin{align*}
    \mathbf{D}= \begin{bmatrix}
           0 & 1 & 0 & 1 \\
           0 & 1 & 1 & 0 \\
           0 & 0 & 1 & 1 \\
         \end{bmatrix}.
  \end{align*}
  The matrix $\mathbf{D}$ is $2$-separable over real field $\mathbb{R}$, but it is not $2$-separable over $\mathbb{F}_2$ since, for instance, the sum of the first and second columns over $\mathbb{F}_2$ is the same as the sum of the third and fourth columns over $\mathbb{F}_2$. 
  \begin{align*}
    \begin{bmatrix}
           0  \\
           0 \\
           0
    \end{bmatrix} \oplus \begin{bmatrix}
           1  \\
           1 \\
           0
    \end{bmatrix}=\begin{bmatrix}
           0  \\
           1 \\
           1
    \end{bmatrix} \oplus \begin{bmatrix}
           1  \\
           0 \\
           1
    \end{bmatrix}=\begin{bmatrix}
           1  \\
           1 \\
           0
    \end{bmatrix}.
  \end{align*}
\end{example}
From the definition, it can be easily seen that if a matrix $\mathbf{D}$ (with $n$ columns) is $t$-separable over a field $\mathbb{F}$, then $\mathbf{D}$ is also $s$-separable over $\mathbb{F}$ for any $1\leq s < t < n$.  

The vector of test results, $\mbf{y}$, is the sum of the columns in the testing matrix corresponding to the coordinates of the defective items. When a $t$-separable matrix over $\mathbb{R}$ is used as the testing matrix, the vector $\mbf{y}$ will be distinct for any set of $t$ defective items. Thus, a $t$-separable matrix over $\mathbb{R}$ can be used as the testing matrix for identifying $t$ defective items. However, the construction of {$t$-separable} matrices for arbitrary $t$ with minimum number of rows is an open problem. Instead, we can leverage the idea that the parity-check matrix of any binary  $t$-error-correcting code is a $t$-separable matrix over $\mathbb{F}_2$. Note that $t$-separability over $\mathbb{F}_2$ results in $t$-separability over $\mathbb{R}$. Hence, a possible choice for designing a $t$-separable matrix over $\mathbb{R}$ is utilizing the parity-check matrix of a binary $t$-error-correcting code. 

In this work, we use binary BCH codes for this purpose. The key feature of the BCH codes which make them suitable for designing $t$-separable matrices is that it is possible to design binary BCH codes, capable of correcting any combination of $t$ or fewer errors. 
\begin{definition}\label{def:bch}\cite{lin2001error} (Binary BCH codes)
For any positive integers $m\geq 3$ and $t < 2^{m-1}$, there exists a binary $t$-error-correcting BCH code with the following parameters:
\[   
     \begin{cases}
       \text{$n=2^m-1$} &\quad\text{block length}\\
        \text{$n-k \leq mt$} &\quad \text{number of parity-check digits}\\
       \text{$d_{\min}\geq 2t+1$} &\quad\text{minimum Hamming distance}\\
     \end{cases}
\]
The $t\times n$ parity-check matrix of such a code is given by
\[ 
    \mathbf{H}_t=  \begin{bmatrix}
           1 & \alpha & \alpha^2 & \dots &  \alpha^{n-1}\\
           
           1 & (\alpha^3) & (\alpha^3)^2 &  \dots & (\alpha^3)^{n-1}\\
            1 & (\alpha^5) & (\alpha^5)^2 &  \dots &  (\alpha^5)^{n-1}\\
            \vdots& \vdots& \vdots & \ddots & \vdots\\
           1 & (\alpha^{2t-1}) & (\alpha^{2t-1})^2 & \dots &  (\alpha^{2t-1})^{n-1}\\
           
         \end{bmatrix},
  \]
  where $\alpha$ is a primitive element in $\mathbb{F}_{2^m}$.
  \end{definition}
  
  Since each entry of $\mathbf{H}_t$ is an element in $\mathbb{F}_{2^m}$, it can be represented by an $m$-tuple over $\mathbb{F}_2$. Thus, the number of rows in the binary representation of $\mathbf{H}_t$ is given by 
 \begin{equation}\label{eq:rows}
  R=tm=t\log_2 ({n+1}).
  \end{equation}

\subsection{Encoding algorithm}
 
The design of the measurement matrix $\mathbf{A}$ in our scheme is based on an architectural philosophy that was proposed in~\cite{DBLP:journals/corr/LeePR15} and~\cite{DBLP:journals/corr/VemJN17}.
The key idea is to design $\mathbf{A}$ using a sparse bi-regular bipartite graph and to apply a peeling-based iterative algorithm for recovering the defective items given $\mathbf{y}$.

Let $G_{\ell,r} (N,M)$ be a randomly generated bipartite graph where each of the $N$ left nodes is connected to $\ell$ right nodes uniformly at random, and each of the $M$ right nodes is connected to $r$ left nodes uniformly at random. Note that there are $N\ell$ edge connections from the left side and $Mr$ edge connections from the right side,
\begin{equation}\label{eq:connections}
    N\ell=Mr
\end{equation}

Let $\mathbf{T}_{G} \in \{0,1\}^{M\times N}$ be the adjacency matrix of the graph $G_{\ell,r} (N,M)$, where each column in $\mathbf{T}_{\mathcal{G}}$ corresponds to a left node and has exactly $\ell$ ones, and each row corresponds to a right node and has exactly $r$ ones. Let $\mathbf{t}_i \in \{0,1\}^{N}$ denote the $i$-th row of $\mathbf{T}_{G}$, i.e.,  $\mathbf{T}_{G}=[\mathbf{t}_{1}^{\mathsf{T}},\mathbf{t}_{2}^{\mathsf{T}},\cdots,\mathbf{t}_{M}^{\mathsf{T}}]^{\mathsf{T}}$. We assign $s$ tests to each right node based on a signature matrix $\mathbf{U}\in \{0,1\}^{s\times r}$. The signature matrix is chosen as ${\mathbf{U}=[\mathbf{1}_{1\times r}^{\mathsf{T}} ,\mathbf{H}_t^{\mathsf{T}}]^{\mathsf{T}}}$, where $\mathbf{1}_{1\times r} $ is an all-ones row of length $r$, and ${\mathbf{H}_t \in \{0,1\}^{t\log_2(r+1)\times r}}$ is the parity-check matrix of a binary $t$-error-correcting BCH code of length $r$. From~\eqref{eq:rows}, it can be easily seen that ${s=R+1=t\log_2 (r+1)+1}$. 

The measurement matrix is given by ${\mathbf{A}=[\mathbf{A}_{1}^{\mathsf{T}},\cdots,\mathbf{A}_{M}^{\mathsf{T}}]^{\mathsf{T}}}$ where ${\mathbf{A}_i\in \{0,1\}^{s\times N}}$ is a matrix that defines the $s$ tests at the $i$-th right node. There are exactly $r$ ones in each row $\mathbf{t}_i$ of $\mathbf{T}_{G}$, and the signature matrix $\mathbf{U}=[\mathbf{u}_1,\mathbf{u}_2,\cdots,\mathbf{u}_r]$ has $r$ columns. Note that $\mbf{u}_i=[1,\mbf{h}_i^{\mathsf{T}}]^{\mathsf{T}}$ is the $i$-th column of $\mbf{U}$, where $\mbf{h}_i$ is the $i$-th column of $\mbf{H}_t$. $\mathbf{A}_i$ is obtained by placing the $r$ columns of $\mathbf{U}$ at the coordinates of the $r$ ones of the row vector $\mathbf{t}_i$, and replacing zeros by all-zero columns,  
\begin{align}\label{eq:measureblock}
 \mbf{A}_i=[\mbf{0},\ldots,\mbf{0},\mbf{u}_1, \mbf{0},\ldots, \mbf{u}_2,\mbf{0}, \ldots, \mbf{u}_{r}]
 \end{align} where $\mbf{t}_i =[0,\ldots,0,\hspace{0.6ex}1,\hspace{0.9ex} 0, \ldots,\hspace{0.6ex}1,\hspace{0.9ex}0, \ldots, \hspace{0.9ex}1]$.  
 
The number of rows in the measurement matrix $\mathbf{A}$, ${m=M\times s}$ where $s=t\log_2 (r+1)+1$, represents the total number of tests in the proposed scheme.
\begin{example}\label{ex:example1}
Let $N=14$ be the total number of items. Let $G$ be a randomly generated left-and-right-regular graph with $N$ left nodes of degree $\ell=2$ and $M=4$ right nodes of degree $r=7$. For this example, suppose that the adjacency matrix $\mathbf{T}_{G}$ of the graph $G$ is given by
\setcounter{MaxMatrixCols}{14}
\[
\mathbf{T}_{\mathcal{G}}= \begin{bmatrix}
     \textcolor{blue}{1} & 0 & \textcolor{lime}{1} & 0 & \textcolor{orange}{1} & 0 & \textcolor{green}{1} & 0 & \textcolor{red}{1} & 0 & \textcolor{brown}{1} & 0 & 0 & \textcolor{yellow}{1} \\
     0 & \textcolor{blue}{1} & \textcolor{lime}{1} & 0 & 0 & \textcolor{orange}{1} & 0 & \textcolor{green}{1} & 0 & \textcolor{red}{1} & 0 & \textcolor{brown}{1} & 0 & \textcolor{yellow}{1} \\
     0 & \textcolor{blue}{1} & 0 & \textcolor{lime}{1} & 0 & \textcolor{orange}{1} & \textcolor{green}{1} & 0 & 0 & \textcolor{red}{1} & \textcolor{brown}{1} & 0 & \textcolor{yellow}{1} & 0 \\
     \textcolor{blue}{1} & 0 & 0 & \textcolor{lime}{1} & \textcolor{orange}{1} & 0 & 0 & \textcolor{green}{1} & \textcolor{red}{1} & 0 & 0 & \textcolor{brown}{1} & \textcolor{yellow}{1} & 0 
\end{bmatrix}.
\] Consider the parity-check matrix $\mathbf{H}_1$ of a binary $t=1$-error-correcting BCH code of length $r=7$ given by 
 \[
         \mbf{H}_1 =\begin{bmatrix}
         1 & \alpha & \cdots & \alpha^6
         \end{bmatrix}=
         \begin{bmatrix}
0 & 0  & 1 & 0 & 1  & 1 & 1\\
 0  & 1 & 0 & 1  & 1 & 1 & 0\\
1 & 0 & 0 & 1 & 0 & 1 & 1
\end{bmatrix},
\] where $\alpha \in \mathbb{F}_{2^3}$ is a root of the primitive polynomial ${\alpha^3+\alpha+1=0}$. The signature matrix ${\mathbf{U}=[\mathbf{1}_{1\times 7}^{\mathsf{T}},\mathbf{H}_1^{\mathsf{T}}]^{\mathsf{T}}}$ is then given by
\[
\mbf{U} = \begin{bmatrix}
 \textcolor{blue}{1} & \textcolor{lime}{1}  & \textcolor{orange}{1} &  \textcolor{green}{1} & \textcolor{red}{1}  & \textcolor{brown}{1} & \textcolor{yellow}{1}\\
\textcolor{blue}{0} & \textcolor{lime}{0}  & \textcolor{orange}{1} & \textcolor{green}{0} & \textcolor{red}{1}  & \textcolor{brown}{1} & \textcolor{yellow}{1}\\
\textcolor{blue}{0}  & \textcolor{lime}{1} & \textcolor{orange}{0} & \textcolor{green}{1}  & \textcolor{red}{1} & \textcolor{brown}{1} & \textcolor{yellow}{0}\\
\textcolor{blue}{1} & \textcolor{lime}{0} & \textcolor{orange}{0} & \textcolor{green}{1} & \textcolor{red}{0} & \textcolor{brown}{1} & \textcolor{yellow}{1}
\end{bmatrix}.
\] Following the construction procedure explained earlier, the testing matrix $\mathbf{A}$ is then given by
  \[
   \mathbf{A}= \begin{bmatrix}
 \textcolor{blue}{1} & 0 & \textcolor{lime}{1} & 0 & \textcolor{orange}{1} & 0 & \textcolor{green}{1} & 0 & \textcolor{red}{1} & 0 & \textcolor{brown}{1} & 0 & 0 & \textcolor{yellow}{1}  \\
 \textcolor{blue}{0} & 0 & \textcolor{lime}{0} & 0 & \textcolor{orange}{1} & 0 & \textcolor{green}{0} & 0 & \textcolor{red}{1} & 0 & \textcolor{brown}{1} & 0 & 0 & \textcolor{yellow}{1}\\
 \textcolor{blue}{0} & 0 & \textcolor{lime}{1} & 0 & \textcolor{orange}{0} & 0 & \textcolor{green}{1} & 0 & \textcolor{red}{1} & 0 & \textcolor{brown}{1} & 0 & 0 & \textcolor{yellow}{0}\\
 \textcolor{blue}{1} & 0 & \textcolor{lime}{0} & 0 & \textcolor{orange}{0} & 0 & \textcolor{green}{1} & 0 & \textcolor{red}{0} & 0 & \textcolor{brown}{1} & 0 & 0 & \textcolor{yellow}{1}\\ \hline
           
0 & \textcolor{blue}{1} & \textcolor{lime}{1} & 0 & 0 & \textcolor{orange}{1} & 0 & \textcolor{green}{1} & 0 & \textcolor{red}{1} & 0 & \textcolor{brown}{1} & 0 & \textcolor{yellow}{1}\\
0 & \textcolor{blue}{0} & \textcolor{lime}{0} & 0 & 0 & \textcolor{orange}{1} & 0 & \textcolor{green}{0} & 0 & \textcolor{red}{1} & 0 & \textcolor{brown}{1} & 0 & \textcolor{yellow}{1} \\
0 & \textcolor{blue}{0} & \textcolor{lime}{1} & 0 & 0 & \textcolor{orange}{0} & 0 & \textcolor{green}{1} & 0 & \textcolor{red}{1} & 0 & \textcolor{brown}{1} & 0 & \textcolor{yellow}{0}\\
0 & \textcolor{blue}{1} & \textcolor{lime}{0} & 0 & 0 & \textcolor{orange}{0} & 0 & \textcolor{green}{1} & 0 & \textcolor{red}{0} & 0 & \textcolor{brown}{1} & 0 & \textcolor{yellow}{1}\\ \hline 
           
 0 & \textcolor{blue}{1} & 0 & \textcolor{lime}{1} & 0 & \textcolor{orange}{1} & \textcolor{green}{1} & 0 & 0 & \textcolor{red}{1} & \textcolor{brown}{1} & 0 & \textcolor{yellow}{1} & 0\\
 0 & \textcolor{blue}{0} & 0 & \textcolor{lime}{0} & 0 & \textcolor{orange}{1} & \textcolor{green}{0} & 0 & 0 & \textcolor{red}{1} & \textcolor{brown}{1} & 0 & \textcolor{yellow}{1} & 0 \\
 0 & \textcolor{blue}{0} & 0 & \textcolor{lime}{1} & 0 & \textcolor{orange}{0} & \textcolor{green}{1} & 0 & 0 & \textcolor{red}{1} & \textcolor{brown}{1} & 0 & \textcolor{yellow}{0} & 0\\
 0 & \textcolor{blue}{1} & 0 & \textcolor{lime}{0} & 0 & \textcolor{orange}{0} & \textcolor{green}{1} & 0 & 0 & \textcolor{red}{0} & \textcolor{brown}{1} & 0 & \textcolor{yellow}{1} & 0\\\hline
           
\textcolor{blue}{1} & 0 & 0 & \textcolor{lime}{1} & \textcolor{orange}{1} & 0 & 0 & \textcolor{green}{1} & \textcolor{red}{1} & 0 & 0 & \textcolor{brown}{1} & \textcolor{yellow}{1} & 0\\
\textcolor{blue}{0} & 0 & 0 & \textcolor{lime}{0} & \textcolor{orange}{1} & 0 & 0 & \textcolor{green}{0} & \textcolor{red}{1} & 0 & 0 & \textcolor{brown}{1} & \textcolor{yellow}{1} & 0 \\
\textcolor{blue}{0} & 0 & 0 & \textcolor{lime}{1} & \textcolor{orange}{0} & 0 & 0 & \textcolor{green}{1} & \textcolor{red}{1} & 0 & 0 & \textcolor{brown}{1} & \textcolor{yellow}{0} & 0\\
\textcolor{blue}{1} & 0 & 0 & \textcolor{lime}{0} & \textcolor{orange}{0} & 0 & 0 & \textcolor{green}{1} & \textcolor{red}{0} & 0 & 0 & \textcolor{brown}{1} & \textcolor{yellow}{1} & 0\\
         \end{bmatrix}.
\]
\end{example}
\subsection{Decoding algorithm}\label{decoding}
Let the observation vector corresponding to the $i$-th right node be defined as 
\begin{equation}\label{eq:obsvec}
\mathbf{z}_i=[z_{i,1},z_{i,2},\cdots,z_{i,s}]^{\mathsf{T}}=\mathbf{A}_i\mathbf{x},\ \forall i\in\{1,\cdots,M\}.
\end{equation}
Note that  
${\mbf{z}_i=[y_{(i-1)s+1},\cdots,y_{is}]^{\mathsf{T}}}$.
\begin{definition}($t$-\emph{resolvable} right node)
A right node is called $t$-\emph{resolvable} if it is connected to $t$ or fewer defective items.
\end{definition}
The following lemma is useful for resolving the right nodes. (The proofs of all lemmas can be found in the appendix.)
\begin{lemma}\label{lem:complexity}
The proposed algorithm detects and resolves all the $t$-resolvable right nodes.
\end{lemma}

The decoding algorithm performs in rounds as follows. In each round, the decoding algorithm first iterates through all the right node observation vectors $\{\mathbf{z}_i\}_{i=1}^M$, and resolves all $t$-resolvable right nodes (by BCH decoding, as discussed in the proof of Lemma~\ref{lem:complexity}). Then, given the identities of the recovered left nodes, the edges connected to these defective items are peeled off the graph. That is, the contributions of the recovered defective items will be removed from the unresolved right nodes so that new right nodes may become $t$-resolvable for the next round. The decoding algorithm terminates when there is no more $t$-resolvable right nodes.
\begin{example}
Consider the group testing problem in the Example~\ref{ex:example1}. Let the number of defective items be $K=3$ and let $\mathbf{x}=[1,0,0,1,0,0,0,0,0,1,0,0,0,0]^T$, i.e., item $1$, item $4$, and item $10$ are defective items. We show how the proposed scheme can identify the defective items. The result of the tests can be expressed as follows,
\begin{align*}
   \mathbf{y}= \begin{bmatrix}
            \mathbf{z}_1 \\
           \mathbf{z}_2\\
            \mathbf{z}_3 \\
           \mathbf{z}_4\\
    \end{bmatrix}=\mathbf{A}\mathbf{x}=
    \begin{bmatrix}
   \mathbf{u}_1  \\
   \mathbf{u}_{5} \\
   \mathbf{u}_{2}+\mathbf{u}_{5} \\
\mathbf{u}_{1}+\mathbf{u}_{2} \\
    \end{bmatrix}
\end{align*}
Then, the right-node observation vectors are given by
\[ \mathbf{z}_1=\mathbf{u}_1=[1,0,0,1]^{\mathsf{T}}\]
\[ \mathbf{z}_2=\mathbf{u}_{5}=[1,1,1,0]^{\mathsf{T}}\]
\[ \mathbf{z}_3=\mathbf{u}_{2}+\mathbf{u}_{5}=[2,1,2,0]^{\mathsf{T}}\]
\[ \mathbf{z}_4=\mathbf{u}_{1}+\mathbf{u}_{2}=[2,0,1,1]^{\mathsf{T}}\]

 Because the signature matrix is built using a {$1$-separable} matrix, each right node can be resolved if it is connected to at most one defective item.
 
Iteration $1$: we first find the {$1$-resolvable} right nodes. The first and second right nodes are $1$-resolvable because $z_{1,1}=z_{2,1}=1$. Using a BCH decoding algorithm, one can find that the defective items connected to the first and second right nodes are item $1$ and item $10$, respectively. Next, we remove the contributions of the items $1$ and $10$ from the unresolved right nodes. The new observation vectors will be as follows,
\[ \mathbf{z}_3=\mathbf{u}_{2}=[1,0,1,0]^{\mathsf{T}}\]
\[ \mathbf{z}_4=\mathbf{u}_{2}=[1,0,1,0]^{\mathsf{T}}\]

Iteration $2$: it can be easily observed that the third and forth right nodes are {$1$-resolvable} since $z_{3,1}=z_{4,1}=1$. Using a BCH decoding algorithm, it follows that the item $4$ is the defective item connected to both right nodes $3$ and $4$. Since all the $K=3$ defective items are identified, the decoding algorithm terminates. 
\end{example}

\section{Main Results}\label{sec:main}
In this section, we present our main results. Theorem~\ref{theo:main1} characterizes the required number of tests that guarantees the identification of all defective items with probability approaching one as $K,N\rightarrow\infty$. Theorem~\ref{thm:main2} presents the computational complexity of the proposed algorithm. The proofs of Theorems~\ref{theo:main1} and~\ref{thm:main2} are given in Section~\ref{sec:proofs}.

\begin{theorem}\label{theo:main1}
For the sub-linear regime, the proposed scheme recovers all defective items with probability approaching one (as $K,N\rightarrow\infty$) with at most ${m=c(t)K\left(t\log_2\left(\frac{\ell N}{c(t)K}+1\right)+1\right)+1}$ tests, where $c(t)$ depends only on $t$. Table~\ref{Table:constant} shows the values of $c(t)$ for $t\leq 8$. 
\begin{table}[h]
\centering
\scriptsize{
\begin{tabular}{| c | c | c | c |c | c | c | c | c |}
\hline
$t$ & 1 & 2 & 3 & 4 & 5 & 6  & 7 & 8\\ \hline
$c(t)$ & 1.222 & 0.597 & 0.388 & 0.294 & 0.239 & 0.202 & 0.176 & 0.156\\ \hline
 $\ell^{\star}$ & 3 & 2 & 2 & 2 & 2 & 2 & 2 & 2 \\ \hline
\end{tabular}}
\caption{\small{The function $c(t)$ and the optimal left degree $\ell^{\star}$.}}
\label{Table:constant}
\end{table}
\end{theorem}

\begin{theorem}\label{thm:main2}
The encoding and decoding of the proposed algorithm for any $t\leq 4$ have the computational complexity of $\mathcal{O}(K\log^2 \frac{N}{K})$ and $\mathcal{O}(K\log \frac{N}{K})$, respectively.
\end{theorem} 


\section{Proofs of Main Theorems}\label{sec:proofs}
\subsection{Proof of Theorem~\ref{theo:main1}}
Let $N$ be the total number of items, out of which $K$ items are defective. Note that in the QGT problem, performing one initial test (on all items) would suffice to obtain the number of defective items. 
As mentioned in Section~\ref{decoding}, our scheme employs an iterative decoding algorithm. In each iteration, the algorithm finds and resolves all the $t$-\emph{resolvable} right nodes. At the end of each iteration, the decoder subtracts the contribution of the identified defective items from the unresolved right nodes. This process is repeated until there is no $t$-\emph{resolvable} right nodes left in the graph. The fraction of defective items that remain unidentified when the decoding algorithm terminates can be analyzed using density evolution as follows. 

Assuming that the exact number of the defective items, $K$, is known and the values assigned to the defective and non-defective items are one and zero, respectively, the left-and-right-regular bipartite graph can be pruned. All the zero left nodes and their respective edges are removed from the graph. The number of left nodes in the pruned graph is $K$, but the degree of these nodes remains unchanged. On the other hand, the number of right nodes remains unchanged, but the resulting graph is not right-regular any longer. 

Let $\lambda$ be the average right degree, i.e., $\lambda=\frac{K\ell}{M}$. Let $\rho(x)\triangleq \sum_{i=1}^{\min(K,r)}\rho_ix^{i-1}$ be the right edge degree distribution, where $\rho_i$ is the probability that a randomly picked edge in the pruned graph is connected to a right node of degree $i$, and $\min(K,r)$ is the maximum degree of a right node. As shown in~\cite{DBLP:journals/corr/VemJN17}, as $K,N\rightarrow\infty$, we have $\rho_i=e^{-\lambda}\frac{\lambda^{i-1}}{(i-1)!}$.
 
The following lemma is useful for computing the fraction of unidentified defective items at each iteration $j$ of the decoding algorithm.
 \begin{lemma}\label{lem:DenEvol}
 Let $p_j$ be the probability that a randomly chosen defective item is not recovered at iteration $j$ of the decoding algorithm; and let $q_j$ be the probability that a randomly picked right node is resolved at iteration $j$ of the decoding algorithm. The relation between $p_j$ and $p_{j+1}$ is determined by the following density evolution equations: 
\begin{equation}\label{eq:det-denevol1}
q_j=\sum_{i=1}^{t}\rho_i+
    \sum_{i=t+1}^{\min(K,r)}\rho_i\sum_{k=0}^{t-1}{i-1 \choose k}p_j^k(1-p_j)^{i-k-1},
    \end{equation}
    \begin{equation}\label{eq:det-denevol2}
    p_{j+1}=(1-q_j)^{\ell-1},
 \end{equation}
where $t$ is the level of separability, and $\rho_i$ is the probability that a randomly picked edge in the pruned graph is connected to a right node of degree $i$. 
 \end{lemma}

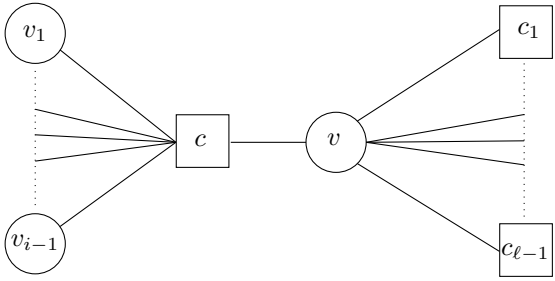
\begin{figure}\hspace{0.5cm}
\begin{tikzpicture}
  \coordinate (O1) at (0,-0.7);
  \coordinate (O2) at (0,-3.5);
  \coordinate (O3) at (4,-2.15);
  \draw (O1) circle (0.4);
  \draw (O2) circle (0.4);
  \draw (O3) circle (0.4);
  \draw[black] (2.2,-1.9) ++(-.65*0.5, -7/6*0.5) rectangle ++(0.7,0.7);
   \draw[black] (6.5,-0.45) ++(-.65*0.5, -7/6*0.5) rectangle ++(0.7,0.7);
   \draw[black] (6.5,-3.35) ++(-.65*0.5, -7/6*0.5) rectangle ++(0.7,0.7);
  \draw [dotted] (0,-1.2) -- (0,-3);
  \draw [dotted] (6.5,-1.1) -- (6.5,-3.2);
  \path[black] (-90:0.7) node []{$v_1$};
  \path[black] (-90:3.5) node []{$v_{i-1}$};
  \path[black] (-44:3.05) node []{$c$};
  \path[black] (-28:4.5) node []{$v$};
  \path[black] (-5.75:6.6) node []{$c_1$};
  \path[black] (-29:7.45) node []{$c_{\ell-1}$};
  \draw (0.33,-0.92) -- (1.87,-2.15);
  \draw (0.33,-3.27) -- (1.87,-2.15);
  
  \draw (0,-1.71) -- (1.87,-2.15);
  \draw (0,-2.05) -- (1.87,-2.15);
  \draw (0,-2.4) -- (1.87,-2.15);

  \draw (4.4,-2.15) -- (6.5,-2.13);
  \draw (4.4,-2.15) -- (6.5,-2.45);
  \draw (4.4,-2.15) -- (6.5,-1.78);

  \draw (2.6,-2.15) -- (3.6,-2.15);
  \draw (4.283,-1.867) -- (6.17,-0.65);
  \draw (4.283,-2.433) -- (6.17,-3.6);
\end{tikzpicture}
\caption{Tree-like representation of neighborhood of the edge between a left node $v$ and a right node $c$ in the pruned graph.}\label{fig:denevol}
\end{figure}

Note that $p_j$ is only a function of the variables $t$, $\ell$, and $\lambda$ when $\min(K,r)\rightarrow \infty$. Recall that the goal is to minimize the total number of tests, i.e., ${M\times s}$, where $M$ is the number of right nodes, and $s$ is the number of rows in the signature matrix. The number of rows, $s$, in the signature matrix depends only on the level of separability, $t$. For a given $t$, we can minimize the number of right nodes ${M=\frac{\ell}{\lambda}K}$ subject to the constraint ${\lim_{j\rightarrow \infty }p_j(\ell,\lambda)=0}$, so as to minimize the total number of the tests. The constraint $ \lim_{j\rightarrow \infty }p_j(\ell,\lambda)=0$ guarantees that running the decoding algorithm for sufficiently large number of iterations, the probability that a randomly chosen defective item remains unidentified approaches zero. For any ${\ell \geq 2}$, let  ${\lambda_T(\ell) \triangleq \sup\{\lambda: \lim_{j\rightarrow \infty }p_j(\ell,\lambda)=0\}}$. Then, for any ${\ell \geq 2}$ and $\lambda < \lambda_T(\ell)$, we have ${\lim_{j\rightarrow \infty }p_j(\ell,\lambda)=0}$. Accordingly, for any ${\ell \geq 2}$ and $M=\frac{\ell}{\lambda}K > \frac{\ell}{\lambda_T(\ell)}K$, it follows that $\lim_{j\rightarrow \infty }p_j(\ell,\lambda)=0$. Our goal is then to compute

\begin{equation}\label{eq:opt}
    \min_{\ell \in \{2,3,\dots\}}\frac{\ell}{\lambda_T(\ell)}K.
\end{equation}
We can solve this problem numerically and attain the optimal value of $\ell$, i.e., $\ell^{\star}$. Let ${c(t)\triangleq \frac{\ell^{\star}}{\lambda_T(\ell^{\star})}}$. The number of right nodes can then be chosen as ${M=c(t)K\beta}$ for any $\beta> 1$ to guarantee that ${M > c(t)K=\frac{\ell^{\star}}{\lambda_T(\ell^{\star})}K}$. Substituting ${M=c(t)K\beta}$ in \eqref{eq:connections} results in ${r=\frac{\ell N}{c(t)K\beta}}$. Therefore, the total number of tests will become ${M\times s=c(t)K\beta\left(t\log_2\left(\frac{\ell N}{c(t)K\beta}+1\right)+1\right)}$. 
\begin{lemma}\label{lem:lem3}
There exist some $\beta>1$ such that 
\begin{multline*}
 c(t)K\left(t\log_2\left(\frac{\ell N}{c(t)K}+1\right)+1\right)+1 \geq \\ c(t)K\beta\left(t\log_2\left(\frac{\ell N}{c(t)K\beta}+1\right)+1\right).
\end{multline*}
\end{lemma}

By combining the result of Lemma~\ref{lem:lem3} and the preceding arguments, it follows that with probability approaching one as $K,N\rightarrow\infty$,  $m=c(t)K\left(t\log_2\left(\frac{\ell N}{c(t)K}+1\right)+1\right)+1$ tests would suffice for the proposed algorithm to recover all defective items. This completes the proof. 

\subsection{Proof of Theorem~\ref{thm:main2}}
\begin{lemma}\label{lem:complexity2}
For any $t\leq 4$, the computational complexity of resolving each $t$-resolvable right node is $\mathcal{O}(\log r)$.
\end{lemma}  

The total number of right nodes, $M$, is $\mathcal{O}(K)$. From Lemma~\ref{lem:complexity2}, it then follows that the complexity of the decoding algorithm is $\mathcal{O}(K\log r)$. Using \eqref{eq:connections}, it is easy to see that for any $t\leq 4$ the decoding algorithm has complexity $\mathcal{O}(K\log \frac{N}{K})$. The total number of measurements is $m$ and for each measurement $r$ summations are performed. Hence, the complexity of the encoding algorithm is $\mathcal{O}(mr)$, which becomes equivalent to $\mathcal{O}(K\log^2 \frac{N}{K})$ for any $t\leq 4$.

\section{Evaluation of $c(t)$}
In this section, we present the complete analysis for the case of $t=1$, and show how one can evaluate $c(t)$ at $t=1$, i.e., $c(1)$. The same procedure can be used for evaluating $c(t)$ at any $t>1$. 

To compute ${c(1)= \frac{\ell^{\star}}{\lambda_T(\ell^{\star})}}$, we compute the ratio $\frac{\ell}{\lambda_T(\ell)}$ for each $\ell\geq 2$ and its corresponding $\lambda_T(\ell)$. The optimal $\ell$, i.e., $\ell^{\star}$, is the one that yields the minimum value for $\frac{\ell}{\lambda_T(\ell)}$. 

For the case of $t=1$, the density evolution equations~\eqref{eq:det-denevol1} and~\eqref{eq:det-denevol2} can be combined as \begin{equation}\label{eq:pj1} p_{j+1}=\left(1-\sum_{i=1}^{\min(K,r)}\rho_i(1-p_j)^{i-1}\right)^{\ell-1}.
\end{equation} Obviously, $p_1=1$. Substituting ${\rho_i=e^{-\lambda}\frac{\lambda^{i-1}}{(i-1)!}}$, we can rewrite~\eqref{eq:pj1} as
\begin{equation}\label{eq:pj12}
p_{j+1}=\left(1-e^{-\lambda}\sum_{i=1}^{\min(K,r)}\frac{\lambda^{i-1}}{(i-1)!}(1-p_j)^{i-1}\right)^{\ell-1}.
\end{equation} For the sub-linear regime, $\frac{K}{N}\rightarrow 0$ (by definition) as $K,N\rightarrow \infty$, and hence, $r\rightarrow \infty$ (by~\eqref{eq:connections}). Thus, in the asymptotic regime of our interest, $\min(K,r)\rightarrow \infty$. Letting $\min(K,r)\rightarrow \infty$, the equation~\eqref{eq:pj12} reduces to
\begin{equation}\label{eq:t=1}
p_{j+1}=\left(1-e^{-\lambda p_j}\right)^{\ell-1}.
\end{equation}
Using~\eqref{eq:t=1}, we can write
\[
    {\lambda=  \left( \frac{\ln \left(1-p_{j+1}^{\frac{1}{\ell-1}}\right)}{-p_j}\right)}.
\]

The following two lemmas are useful for computing $\lambda_T(\ell) = \sup\{\lambda: \lim_{j\rightarrow \infty }p_j(\ell,\lambda)=0\}$ for each $\ell \geq 2$.
\begin{lemma}\label{lem:convergence}
For any $\ell \geq 2$ and any $\lambda > 0$, the infinite sequence $\{p_1,p_2,\cdots\}$ converges.
\end{lemma}


\begin{lemma}\label{lem:limit}
Let $p^{*}$ be the limit of the sequence $\{p_1,p_2,\cdots\}$, and let \[{\lambda_T(\ell)\triangleq \displaystyle \inf_{0<x<1} \left( \frac{\ln(1-x^{\frac{1}{\ell-1}})}{-x}\right)}.\] Then, for any $\ell \geq 2$, we have 
\[
\begin{cases}
p^{*}=0, & \quad  0 < \lambda < \lambda_T(\ell),\\
p^{*}>0, &\quad \lambda \geq \lambda_T(\ell).\\
\end{cases}
\]
\end{lemma}

By the result of Lemma~\ref{lem:limit}, for any ${\ell\geq 2}$ the value of $\lambda_T(\ell)$ can be computed numerically. One can then obtain the optimal value of $\ell$, i.e., $\ell^{\star}$, which minimizes the ratio of $\frac{\ell}{\lambda_T(\ell)}$, and accordingly ${c(1)= \frac{\ell^{\star}}{\lambda_T(\ell^{\star})}}$ can be computed.

\begin{figure}
\includegraphics[width=0.52\textwidth]{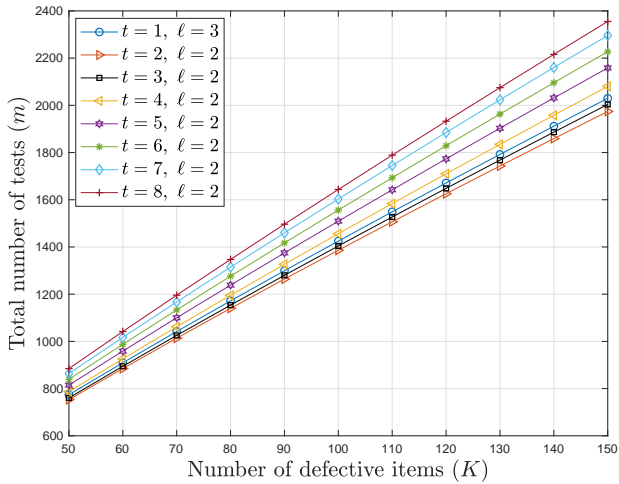}\vspace{-0.2cm}	
\caption{The number of required tests ($m$) to identify all $K$ defective items (for different values of $K$) among $N=2^{16}$ items for different values of ${t}$ obtained via analysis. 
}
\label{fig:best}
\end{figure}

\section{Comparison Results}
In this section we will evaluate the performance of the proposed algorithm based on our theoretical analysis and the Monte Carlo simulations. 

Based on the results in Theorem~\ref{theo:main1} and Table~\ref{Table:constant}, Fig.~\ref{fig:best} depicts the total number of tests ($m$) required to identify all the defective items for different values of $t$. The number of items is assumed to be $N=2^{16}$. As it can be seen, when $t\in \{1,2,3\}$ the required number of tests for identifying all the defective items is less than that for larger values of $t$.

Using the Monte Carlo simulation, we also compare the performance of the proposed scheme for $t\in \{1,2,3\}$ with the performance of the Multi-Level Group Testing (MLGT) algorithm from \cite{8335478}. The MLGT scheme is a semi-quantitative group testing scheme where the result of each test is an integer in the set $\{0,1,2,\cdots,L\}$. Letting ${L\rightarrow \infty}$, the MLGT scheme becomes a QGT scheme. Based on the optimization that we have performed, the optimal left degree for the MLGT scheme is $\ell^{\star}=3$ when $L\rightarrow \infty$. For $K=100$ defective items among a population of $N=2^{16}$ items, the average fraction of unidentified defective items for the MLGT scheme and the proposed scheme are shown in Fig.~\ref{fig:sim} for different values of $m/K$. As it can be observed, the proposed scheme for all the three tested values of $t$ outperforms the MLGT scheme significantly. For instance, when the fraction of unidentified defective items is $2\times 10^{-4}$, the required number of tests for the MLGT scheme (for $\ell=3$) is $3$ times, $5$ times, and $7$ times more than that of the proposed scheme for $t=1$, $t=2$, and $t=3$, respectively.

\begin{figure}
\includegraphics[width=0.52\textwidth]{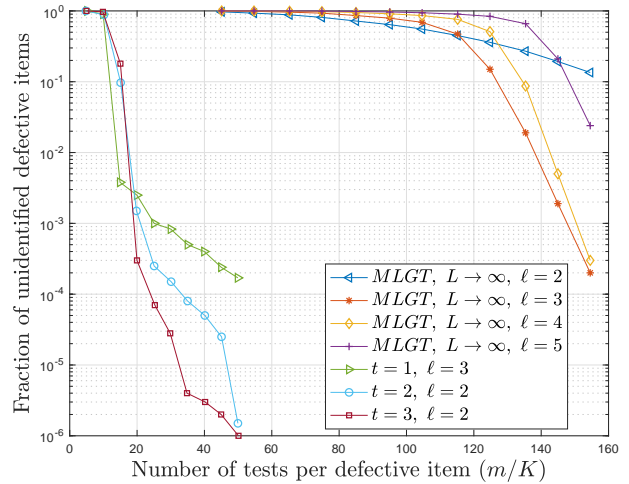}\vspace{-0.2cm}	
\caption{{The average fraction of unidentified defective items obtained via Monte Carlo simulations for $N=2^{16}$ items among which $K=100$ items are defective.}}\label{fig:sim}
\end{figure}

\bibliographystyle{IEEEtran}
\bibliography{QGTRefs}

\appendix[Proof of Lemmas]

\begin{proof}[Proof of Lemma~\ref{lem:complexity}]
Let us divide $\mbf{z}_i$ into two blocks, ${\mbf{z}_i=[{\mbf{z}_{i}^{(1)}}^{\mathsf{T}},{\mbf{z}_{i}^{(2)}}^{\mathsf{T}}]^{\mathsf{T}}}$, where $\mbf{z}_{i}^{(1)}=z_{i,1}$ and ${\mbf{z}_{i}^{(2)}=[z_{i,2},\cdots,z_{i,s}]^{\mathsf{T}}}$. We can rewrite \eqref{eq:obsvec} by placing $[1,\mbf{h}_i^{\mathsf{T}}]^{\mathsf{T}}$ at the coordinates of $\mbf{u}_i$'s in \eqref{eq:measureblock},
\begin{align*}\label{eq:deceq}
 \begin{bmatrix} \mbf{z}_{i}^{(1)}\\\mbf{z}_{i}^{(2)}
  \end{bmatrix}&= \begin{bmatrix}
  0& \ldots&0&1& 0&\ldots& 1&0& \ldots& 1\\
 \mbf{0}&\ldots&\mbf{0}&\mbf{h}_1& \mbf{0}&\ldots&\mbf{h}_2&\mbf{0}& \ldots&\mbf{h}_r
 \end{bmatrix}\mbf{x}.
\end{align*}
Assume that $j\leq t$ defective items are connected to the $i$-th right node. The first block, $\mbf{z}_{i}^{(1)}$, which is the first element of $\mbf{z}_i$, shows the number of defective items connected to the $i$-th right node. Recall that the first row of the signature matrix is an all-ones vector. It means that there are $r$ ones in the first row of every $\mathbf{A}_i$, $i\in\{1,2,\cdots,M\}$. Thus, all $r$ items connected to the $i$-th right node are included in the test corresponding to the first row of $\mathbf{A}_i$. The second block, $\mbf{z}_{i}^{(2)}$, is equal to the sum of $\mbf{h}_i$'s corresponding to the defective items connected to the $i$-th right node. Let $S_i$ be the set of indices of items (left nodes) that are connected to the $i$-th right node, and let $\mathbf{x}_{S_i}$ be the vector $\mathbf{x}$ restricted to the items indexed by $S_i$. Note that $\mathbf{x}_{S_i}$ can be viewed as an error vector for a $t$-error-correcting BCH code with parity-check matrix $\mathbf{H}_t$, and the block vector $\mbf{z}_i^{(2)}$ under modulo $2$ can be interpreted as the syndrome corresponding to the error vector $\mathbf{x}_{S_i}$. The Hamming weight of the error vector $X_{S_i}$, i.e., the number of ones in $\mathbf{x}_{S_i}$, is equal to $j$. When $j\leq t$, the error vector $\mathbf{x}_{S_i}$ can be decoded from the corresponding syndrome by decoding the underlying BCH code, and hence all $j$ defective items connected to the $i$-th right node can be identified. 
\end{proof}

\begin{proof}[Proof of Lemma~\ref{lem:DenEvol}]
 As mentioned earlier, the pruned graph is left-regular and the degree of the left nodes is $\ell$, but the pruned graph is not right-regular any longer and the degree of the right nodes can be any integer in $\{0,1,\cdots,\min(K,r)\}$. A tree-like representation of the neighborhood of an edge between a left node $v$ of degree $\ell$ and a right node $c$ of degree $i$ is shown in Fig.~\ref{fig:denevol}. The left node $v$ sends a \textquote{not identified} message to the right node $c$ at iteration $j+1$ with probability $p_{j+1}$ if all of its neighboring nodes $\{c_i\}_{i=1}^{\ell-1}$ have not been resolved at iteration $j$ which it happens with probability $(1-q_j)^{l-1}$. The right node $c$ of degree $i$ with probability $q_j$ passes a \textquote{resolved} message to the left node $v$ at iteration $j$ if the number of defective items connected to node $c$, i.e., $i$, is equal to $t$ or less which it happens with probability $\sum_{i=1}^{t}\rho_i$, or if the number of defective items connected to node $c$ is more than $t$ ($i>t$), but only ${k\in\{0,1,\cdots,t-1\}}$ of the $i-1$ defective items connected to node $c$ other than $v$ are unidentified (we know that $v$ is not identified yet) which this case happens with probability $\sum_{i=t+1}^{\min(K,r)}\rho_i\sum_{k=0}^{t-1}{i-1 \choose k}p_j^k(1-p_j)^{i-k-1}$. 
\end{proof}

\begin{proof}[Proof of Lemma~\ref{lem:lem3}]
Let us define the following function,
\[{f(\beta)\triangleq c(t)K\left(t\log\left(\frac{\ell N}{c(t)K\beta}+1\right)+1\right)}.\] 
We need to show that there exists some $\beta>1$ such that ${f(1)+1\geq \beta f(\beta)}$, or equivalently, $\beta f(\beta)-f(1) \leq 1$. Since $f(\beta)$ is a monotone decreasing function of $\beta$, $f(\beta)<f(1)$ for $\beta>1$. This inequality leads to ${\beta f(\beta)-f(1) < (\beta-1)f(1)}$. Hence, to guarantee that there exists some $\beta>1$ such that $\beta f(\beta)-f(1) \leq1$, it suffices to show that $(\beta-1)f(1)\leq1$ for some $\beta>1$. It is easy to see that $1<\beta\leq \frac{1}{f(1)}+1 $ is the satisfactory range. 
\end{proof}

\begin{proof}[Proof of Lemma~\ref{lem:complexity2}]
As mentioned in Lemma~\ref{lem:complexity}, the block vector $\mbf{z}_i^{(2)}$ under modulo $2$ can be interpreted as the syndrome corresponding to an error pattern of Hamming weight $j \leq t$. The location of the $j$ errors ($j$ defective items) can be determined from $\mbf{z}_i^{(2)}$ under modulo $2$ by first using a Berlekamp-Massey algorithm for finding the error locator polynomial. This step involves a time complexity of $\mathcal{O}(t^2 \log r)$ (all computations are performed in a finite field of size $2^m=r+1$). Once the error locator polynomial is determined, the roots of the error locator polynomial have to be found. A standard Chien search can be used to solve this step with complexity $\mathcal{O}(tr \log r)$; however, when $t \leq 4$, the Chien search can be avoided and the roots can be found directly using the algorithm in~\cite{chen1982formulas} with a complexity that is only $\mathcal{O}(t\log r)$. Therefore, for $t \leq 4$, the decoding complexity of resolving a $t$-resolvable right node is only logarithmic in $r$ (i.e., $\mathcal{O}(\log r)$).
\end{proof}

\begin{proof}[Proof of Lemma~\ref{lem:convergence}]
 Note that every bounded and monotonic sequence converges. From the definition, it is obvious that ${0 \leq p_{j} \leq 1}$ for any integer $\ell \geq 2$ and any real number $\lambda > 0$. Then, it suffices to show the monotonicity of the sequence $\{p_1,p_2,\dots\}$. The proof is based on induction. It is easy to see that ${p_2 < p_1}$, i.e., $\left(1-e^{-\lambda}\right)^{\ell-1} < 1$. The induction hypothesis is that ${p_j< p_{j-1}}$. We need to show that $p_{j+1}< p_{j}$. By the induction hypothesis, we have \[{\left(1-e^{-\lambda p_{j-1}}\right)^{\ell-1} < p_{j-1}}.\] Then, it is easy to see that \[1-e^{-\lambda \left(1-e^{-\lambda p_{j-1}}\right)^{\ell-1}} < 1-e^{-\lambda p_{j-1}},\] or equivalently, 
 \begin{equation}\label{eq:Dens1sep}
 \left(1-e^{-\lambda \left(1-e^{-\lambda p_{j-1}}\right)^{\ell-1}}\right)^{\ell-1} < \left(1-e^{-\lambda p_{j-1}}\right)^{\ell-1}
 \end{equation}
 Replacing $\left(1-e^{-\lambda p_{j-1}}\right)^{\ell-1}$ by $p_{j}$, we can rewrite \eqref{eq:Dens1sep} as \[\left(1-e^{-\lambda p_{j}}\right)^{\ell-1} < p_{j},\] which yields $p_{j+1}<p_j$, as was to be shown.
\end{proof}

\begin{proof}[Proof of Lemma~\ref{lem:limit}]
By Lemma~\ref{lem:convergence}, we know that $p^{*}$ exists, and it must be a solution to the following equation,
\begin{equation}\label{eq:pinfty}
 p^{*}=\left(1-e^{-\lambda p^{*}}\right)^{\ell-1}.   
\end{equation}

We first show that for ${0 < \lambda < \lambda_T(\ell)}$, it holds that $p^{*}=0$. It suffices to show that for ${0 < \lambda < \lambda_T(\ell)}$ and any integer $\ell\geq 2$, the only solution of \eqref{eq:pinfty} is $p^{*}=0$. Obviously, ${p^{*}=0}$ is a solution of \eqref{eq:pinfty} for any $0 < \lambda < \lambda_T(\ell)$ and any integer $\ell \geq 2$. Thus, we need to show that for $0 < \lambda < \lambda_T(\ell)$ and any integer $\ell\geq 2$, and any $0< \epsilon < 1$, we have ${\epsilon\neq \left(1-e^{-\lambda \epsilon}\right)^{\ell-1}}$. The proof is by the way of contradiction. Suppose that ${\epsilon= \left(1-e^{-\lambda \epsilon}\right)^{\ell-1}}$ for some $0<\epsilon < 1$. By solving this equation for $\lambda$, we get \[\lambda=  \frac{\ln(1-\epsilon^{\frac{1}{\ell-1}})}{-\epsilon}.\] On the other hand, we know that \[\lambda < \lambda_T(\ell)=\displaystyle \inf_{0<x<1} \left( \frac{\ln(1-x^{\frac{1}{\ell-1}})}{-x}\right).\] Thus, we have \[\frac{\ln(1-\epsilon^{\frac{1}{\ell-1}})}{-\epsilon} < \displaystyle \inf_{0<x<1} \left( \frac{\ln(1-x^{\frac{1}{\ell-1}})}{-x}\right)\] for some ${0<\epsilon < 1}$. Obviously, this inequality cannot hold, and we reach a contradiction, as desired.

Next, we shall show that for any ${\lambda\geq \lambda_T(\ell)}$, we have ${p^{*}> 0}$. From \eqref{eq:pinfty}, it follows that \[\lambda= \frac{\ln(1-{p^{*}}^{\frac{1}{\ell-1}})}{-p^{*}}.\] Hence, $\lambda\geq \lambda_T(\ell)$ implies that \[ \frac{\ln(1-{p^{*}}^{\frac{1}{\ell-1}})}{-p^{*}} \geq \displaystyle \inf_{0<x<1} \left( \frac{\ln(1-x^{\frac{1}{\ell-1}})}{-x}\right).\] Again, the proof is by the way of contradiction. Suppose that ${p^{*}=0}$, i.e., the sequence $\{p_1,p_2,\dots\}$ converges to $0$. Therefore, for any ${\delta > 0}$, there exist a positive integer $i$ such that for any $j\geq i$, $\abs{p^{*}-p_j}=p_j < \delta$. Consider an arbitrary $0<\delta<1$. Let $i$ be such that $p_{i-1}\geq \delta$ and $p_j<\delta$ for all $j\geq i$. Note that $p_i < \delta$ implies that $\left(1-e^{-\lambda p_{i-1}}\right)^{\ell-1} < \delta$. This inequality can be rewritten as \[{\lambda< \frac{\ln(1-\delta^{\frac{1}{\ell-1}})}{-p_{i-1}}}.\] Using the facts that $\lambda\geq \lambda_T(\ell)$ and $p_{i-1}\geq \delta$, we have
\begin{equation}\label{eq:14}
\displaystyle \inf_{0<x<1} \left( \frac{\ln(1-x^{\frac{1}{\ell-1}})}{-x}\right) < \frac{\ln(1-\delta^{\frac{1}{\ell-1}})}{-p_{i-1}} 
\end{equation}
\begin{equation}\label{eq:15}
 \frac{\ln(1-\delta^{\frac{1}{\ell-1}})}{-p_{i-1}} \leq \frac{\ln(1-\delta^{\frac{1}{\ell-1}})}{-\delta}
\end{equation}
Combining~\eqref{eq:14} and~\eqref{eq:15}, we get 
\begin{equation}\label{eq:inequality}
\displaystyle \inf_{0<x<1} \left( \frac{\ln(1-x^{\frac{1}{\ell-1}})}{-x}\right) < \frac{\ln(1-\delta^{\frac{1}{\ell-1}})}{-\delta}.
\end{equation} Let $f(x)\triangleq \frac{\ln(1-x^{\frac{1}{\ell-1}})}{-x}$, and let $x^{*}$ be such that \[\inf_{0<x<1}f(x)= \frac{\ln(1-{x^{*}}^{\frac{1}{\ell-1}})}{-x^{*}}.\] Since ${\lim_{x\rightarrow 0} f(x)=\lim_{x\rightarrow 1}f(x)=+\infty}$, obviously we have ${0<x^{*}<1}$. Taking $\delta=x^{*}$, we will have
\begin{equation}\label{eq:equality}
  \displaystyle \inf_{0<x<1} \left( \frac{\ln(1-x^{\frac{1}{\ell-1}})}{-x}\right) = \frac{\ln(1-\delta^{\frac{1}{\ell-1}})}{-\delta}.
\end{equation}
From \eqref{eq:inequality} and \eqref{eq:equality}, we arrive at a contradiction. This completes the proof. 
\end{proof}
\end{document}